\documentclass[11pt]{amsart}
\usepackage{amsmath}
\usepackage{amsfonts}

\usepackage{amscd}
\usepackage{amsthm}
\usepackage{mathrsfs}
\usepackage{amssymb} \usepackage{latexsym}
\usepackage{eufrak}
\usepackage{euscript}
\usepackage{epsfig}
\usepackage{graphics}
\usepackage{array}
\usepackage{enumerate}
\usepackage{dsfont}
\usepackage{color}
\usepackage{wasysym}
\usepackage{hyperref}
\usepackage{pdfsync}

\numberwithin{equation}{section}

\newcommand{\bel}[1]{\begin{equation}\label{#1}}

\newcommand{\be}{\begin{equation}}

\newcommand{\ba}{\begin{eqnarray}}
\newcommand{\ea}{\end{eqnarray}}

\newcommand{\qe}{\end{equation}}
\newcommand{\R}{{\mathbb R}}

\newcommand{\N}{{\mathbb N}}

\newcommand{\Hmm}[1]{\leavevmode{\marginpar{\tiny%
$\hbox to 0mm{\hspace*{-0.5mm}$\leftarrow$\hss}%
\vcenter{\vrule depth 0.1mm height 0.1mm width \the\marginparwidth}%
\hbox to
0mm{\hss$\rightarrow$\hspace*{-0.5mm}}$\\\relax\raggedright #1}}}

\newtheorem{theorem}{Theorem}[section]
\newtheorem{lemma}[theorem]{Lemma}
\newtheorem{coro}[theorem]{Corollary}
\newtheorem{defi}[theorem]{Definition}

\newtheorem{rem}[theorem]{Remark}

\newtheorem{prop}[theorem]{Proposition}

\newtheorem{example}[theorem]{Example}

\begin{document}

\title[Cheeger estimates of Dirichlet-to-Neumann operators]{Cheeger estimates of Dirichlet-to-Neumann operators on infinite subgraphs of graphs}

\author{Bobo Hua}
\address{School of Mathematical Sciences, LMNS, Fudan University, Shanghai 200433,
China; Shanghai Center for Mathematical Sciences, Fudan University, Shanghai 200433, China}
\email{bobohua@fudan.edu.cn}

\author{Yan Huang}
\address{School of Mathematical Sciences, Fudan University, Shanghai 200433, China.}
\email{yanhuang0509@gmail.com}

\author{Zuoqin Wang}
\address{School of Mathematical Sciences, University of Science and Technology of China, Hefei,
Anhui 230026, China.}
\email{wangzuoq@ustc.edu.cn}

\thanks{B. H. is supported in part by NSFC, No.11831004. Z. W. is supported in part by NSFC, Nos. 11571131 and 11721101. B. H. and Z. W. are supported by NSFC no. 11826031.}

\begin{abstract}
In this paper, we study the Dirichlet-to-Neumann operators on infinite subgraphs of graphs. For an infinite graph, we prove Cheeger-type estimates for the bottom spectrum of the Dirichlet-to-Neumann operator, and the higher order Cheeger estimates for higher order eigenvalues of the Dirichlet-to-Neumann operator.
\end{abstract}

\maketitle

\section{introduction}

Eigenvalue estimates are of interest in Riemannian geometry and mathematical physics. There are various eigenvalue estimates using geometric quantities. In this paper, we focus on the isoperimetric type estimate introduced by Cheeger \cite{Cheeger1970}, now called the Cheeger estimate, which reveals a close relation between the first nontrivial eigenvalue of the Laplace-Beltrami operator on a closed manifold, and the isoperimetric constant called Cheeger constant.

Let $(M,g)$ be a compact, connected, smooth Riemannian manifold with smooth boundary $\partial{M}$. The Dirichlet-to-Neumann operator $\Lambda$, which will be referred as the DtN operator for short,
is defined as
\[\aligned
\Lambda:H^{\frac{1}{2}}(\partial M) & \to H^{-\frac{1}{2}}(\partial M),
\\
f &\mapsto\Lambda(f):=\frac{\partial u_f}{\partial n},
\endaligned\]
where $u_f$ is the harmonic extension of $f\in H^{\frac{1}{2}}(\partial M)$. The DtN operator $\Lambda$ is a first order elliptic pseudo-differential operator \cite[p.37]{Taylor1996}. Since $\partial M$ is compact, the spectrum of $\Lambda$ is nonnegative, discrete and unbounded \cite[p. 95]{Bandle1980}. We refer to \cite{Girouard2017} for a survey of the spectral properties of the DtN operators. The eigenvalue problem associated to the DtN operator $\Lambda$ is also known as the Steklov problem. For the history of this problem, c.f. \cite{Kuznetsov2014}. There are many results on the Steklov problem, see e.g.  \cite{Escobar1997,FraserSchoen2011,ColboisEl2011,Hassannezhad2011,Escobar1999,FraserSchoen2016,Jammes2015}.

For the first nontrivial eigenvalue of the DtN operator, Jammes \cite{Jammes2015} introduced a type of Cheeger constant,
$$h_J(M):=\inf_{\substack{\Omega\subset M\\\mathrm{Vol}(\Omega)\leq\frac 12{\mathrm{Vol}(M)}}}\frac{\mathrm{Area}(\partial\Omega\cap \mathrm{int}(M))}{\mathrm{Area}(\Omega\cap\partial M)},$$
where $\mathrm{Area}(\cdot)$ and $\mathrm{Vol}(\cdot)$ denote the codimension one Hausdorff measure (i.e., the Riemannian area), and the Riemannian volume, respectively.
We call $h_J(M)$ the Jammes-type Cheeger constant and the subscript ``$J$" indicates the Jammes-type.  Let $\sigma_2(M)$ be the first nontrivial eigenvalue of $\Lambda$. The Jammes-type Cheeger estimate \cite[Theorem~1]{Jammes2015} reads as
\begin{eqnarray}\label{jammescheegerestimate}\sigma_2(M)\geq\frac{1}{4}h_{N}(M) h_J(M),\end{eqnarray}
where $h_{N}(M)$ is the Neumann Cheeger constant associated to the Neumann Laplacian on $M$.

There are many interesting interactions between Riemannian geometry and discrete analysis on graphs.
Many methods initiated in Riemannian geometry have been generalized to the discrete setting, and conversely some approaches found on graphs may also be applied to Riemannian geometry. The Cheeger estimate was first generalized to graphs by Dodziuk \cite{Dodziuk1984} and Alon-Milman \cite{AlonMilman1985} independently. Miclo introduced higher order Cheeger constants and conjectured related higher order Cheeger estimates, see \cite{Miclo2008,DaneshgarJavadiMiclo2012}. The conjecture was proved by Lee, Oveis Gharan and Tevisan \cite{Trevsan2014} via random partition methods on graphs. Then Miclo \cite{Miclo2015} and Funano \cite{Funano2013} extended the results to the Riemannian case and found some important applications.

Recently, the authors \cite{BoboYanZuoqin2017} defined the DtN operator on a finite subgraph of a graph, and proved two Cheeger-type estimates for the first nontrivial eigenvalue of the DtN operator: the Escobar-type Cheeger estimate following \cite{Escobar1997} and the Jammes-type Cheeger estimate following \cite{Jammes2015}. Hassannezhad-Miclo \cite{Miclo2017} proved the Jammes-type Cheeger estimate independently, and generalized it to the higher order Cheeger estimates for eigenvalues of the DtN operator in terms of higher order Cheeger-Steklov constants. Their result is as follows: Consider a Markov process on a finite state space $U$ and a proper subset $V$ of $U,$ serving as the boundary, on which the DtN operator is defined. Let $\sigma_k$ be the $k$-th eigenvalue of the DtN operator.  Hassannezhad-Miclo \cite[Theorem~A]{Miclo2017} proved that there exists a universal constant $c$ such that
\begin{eqnarray}\label{miclohigherorder}
\sigma_k\geq\frac{c}{k^6}\frac{\varsigma_k}{A}, \quad\forall 1\leq k\leq \sharp V,
\end{eqnarray}
where $A$ is the largest absolute value of the diagonal elements of the irreducible Markov generator and $\varsigma_k$ is the $k$-th order Cheeger-Steklov constant.

Spectral theory for finite graphs has been extensively studied in the literature, see e.g. \cite{Bapat2014,BrouwerHaemers2012,CvetkovicDoobSachs1995,Chung97}. For infinite graphs, there are also many results on the spectra of the Laplacians, see \cite{Mohar1982,MoharWoess1989,Simon2011,BauerHuaJost2014,Fujiwara1996,DodziukMathai2006}. Infinite graphs have many applications in geometric group theory, probability theory, mathematical physics, etc. In this paper, we study the DtN operators on infinite subgraphs of graphs. By the well-known exhaustion methods, see e.g. \cite{BauerHuaJost2014}, we construct the DtN operator on an infinite subgraph of an infinite graph and prove Cheeger-type estimates for the bottom spectrum and higher order Cheeger estimates for the higher order eigenvalues of the DtN operators following \cite{Jammes2015} and \cite{Miclo2017}, which can be viewed as an extension of the results in \cite{BoboYanZuoqin2017} and  \cite{Miclo2017}.

We recall some basic definitions on infinite graphs. Let $V$ be a countable infinite set and $\mu$ be a symmetric weight function given by
\[\aligned
\mu:V\times V & \to [0,\infty),\\
(x,y) &\mapsto \mu_{xy}=\mu_{yx}.
\endaligned\]
This induces a graph structure $G=(V,E)$ with the set of vertices $V$ and the set of edges $E$ such that $\{x,y\}\in E$ if and only if $\mu_{xy}>0$. Two vertices $x,y$ satisfying $\{x,y\}\in E$ are called neighbors, denoted by $x\sim y$.  We only consider locally finite graphs where each vertex only has finitely many neighbors. We do allow self-loops, i.e. $x\sim x$ for $\mu_{xx}>0$. We call the triple $(V,E,\mu)$ a weighted graph.

For an infinite graph, the exhaustion of the whole graph by finite subsets of vertices is an important concept, see \cite{BauerHuaJost2014}. A sequence of subsets of vertices $\mathcal{W}=\{W_i\}^{\infty}_{i=1}$ is called an \emph{exhaustion} of the infinite graph $G=(V,E)$, denoted by $\mathcal{W}\uparrow V$, if it satisfies
\begin{itemize}
\item $W_1\subset W_2\subset\cdots\subset W_i\subset\cdots \subset V$,
\item $\sharp W_i<\infty,\, \forall i=1,2,\cdots$,
\item $V=\bigcup^{\infty}_{i=1} W_i.$
\end{itemize}
For any quantity $\tau$ defined on finite subgraphs of $V$ that is monotone, i.e. for any finite subgraphs $W \subset W'$ of $V$, one has $\tau(W)\leq \tau( W')$ (or $\tau(W)\geq \tau( W')$ as long as $W \subset W'$), we write \begin{equation}\label{exhaustionMono}\tau(V):=\lim_{\mathcal{W}\uparrow V}\tau( W)=\lim_{i\rightarrow\infty}\tau( W_i).\end{equation} One can check that this limit exists and does not depend on the choice of the exhaustion.

Given $\Omega_1,\Omega_2\subset V$, the set of edges between $\Omega_1$ and $\Omega_2$ is denoted by
$$E(\Omega _1,\Omega _2):=\{\{x,y\}\in E|x\in\Omega _1,y\in\Omega _2\}.$$
For any subset $\Omega\subset V$, there are two notions of boundary:
The edge boundary of $\Omega$ is defined as
$$\partial\Omega:=E(\Omega,\Omega^c),$$
where $\Omega^c:=V\setminus\Omega$; the vertex boundary of $\Omega$ is defined as
$$\delta\Omega:=\{x\in \Omega^c|\ x\sim y  \mathrm{\ for\ some\ } y\in \Omega\}.$$

We write $\overline{\Omega}:=\Omega \cup \delta\Omega.$ From now on, we only consider the graph structure $(\overline{\Omega}, E(\Omega,\overline{\Omega}),\mu),$ still denoted by $\overline{\Omega}$ for simplicity, where the weight $\mu$ is modified such that $\mu_{xy}=0$ for any $\{x,y\}\not\in E(\Omega,\overline{\Omega})$, i.e. the edges between vertices in $\delta \Omega,$ $E(\delta\Omega,\delta\Omega)$, are removed. So in what follows we may forget about the ambient graph $(V,E)$ and regard $\overline{\Omega}$ as graph $\Omega$ with its vertex boundary $\delta\Omega.$

We introduce weights on the vertex set $\overline{\Omega}$ as
\[
d(x)=\left\{
       \begin{array}{ll}
       \sum_{y\in\overline{\Omega}}\mu_{xy} ,& \quad x\in\Omega, \\
       \sum_{y\in\Omega}\mu_{xy} ,& \quad x\in\delta\Omega,
               \end{array}
     \right.
\]
which extends to a measure $d(\cdot)$ on subsets of $\overline{\Omega}$ by
$$d(A):=\sum_{x\in A}d(x),\quad \forall\ A\subset\overline{\Omega}.$$

Given a vertex set $F$, we denote by $\R^F$ the collection of all real functions defined on $F$. Let $ W\subset\overline{\Omega}$ be a finite subset and $\delta W$ be the vertex boundary of $ W$ in $\overline{\Omega}$, i.e. $\delta W=\{y\in\overline{\Omega}\setminus W|\exists x\in W\text{\, s.t.\, }x\sim y\}$.
For any $f\in\mathbb{R}^{W\cap\delta\Omega}$, let $u^{ W}_f$ be the solution of the following equation:
\begin{equation}\label{existlemma}
\left\{
\begin{aligned}
\Delta u_{f}^{ W}(x)=0, &\quad x\in W\cap\Omega, \\
u_{f}^{ W}(x)=f(x), &\quad x\in W\cap\delta\Omega,\\
u_{f}^{ W}(x)=0, &\quad x\in\delta W,\\
\end{aligned}
\right.
\end{equation}
see Lemma \ref{harmonic extension}. The third condition above stands for the Dirichlet boundary condition on $ W$.

Now let us define the DtN operator for this infinite subgraph setting.
Let $\ell_0(\delta\Omega)$ denote the set of functions on $\delta\Omega$ with finite support. For any $f\in\ell_0(\delta\Omega)$, we write $f = f^+-f^-$, where $f^+:=\max\{f,0\}$ and $f^-:=\max\{-f,0\}$.
For any $\mathcal{W}\uparrow\overline{\Omega},$ let $u_{f^+}^{ W_i}$ ($u_{f^-}^{ W_i}$ resp.) be the solution of \eqref{existlemma} with $ W$ and $f$ replaced by $ W_i$ and $f^+$ ($f^-$  resp.). Set $$u_f^{ W_i}:=u_{f^+}^{ W_i}-u_{f^-}^{ W_i}.$$
Applying the well-known maximum principle, see e.g. \cite{Grigor'yan2009}, we have $$u_{f^{\pm}}^{ W_i}\leq u_{f^\pm}^{ W_{i+1}}\quad \mathrm{and}\quad |u_{f^{\pm}}^{ W_i}|\leq \sup_{x\in \delta\Omega}|f(x)|.$$ Hence the limit of $u^{ W_i}_f$ exists and we set
$$u_f:=\lim_{i\to \infty}u^{ W_i}_f.$$
For any $f\in\ell_0(\delta\Omega)$, we define
$$\Lambda(f):=\frac{\partial u_f}{\partial n}.$$
From Lemma \ref{welldefined} and Proposition~\ref{prop:selfadjoint}, 
$\Lambda$ is a bounded self-adjoint linear operator on $\ell_0(\delta\Omega)$, which can be uniquely extended to $\ell^2(\delta\Omega)$. We call the extension to $\ell^2(\delta\Omega)$ the DtN operator on $\Omega$ and still denote it by $\Lambda$.
\begin{rem}
Our definition of the DtN operator is slightly different from that of Hassannezhad-Miclo's \cite{Miclo2017}. In fact, the edges between vertices in $\delta \Omega,$ i.e. $E(\delta\Omega,\delta\Omega)$, play no role in our definition, but they matter in Hassannezhad-Miclo's.
\end{rem}
By the standard spectral theory \cite[(4.5.1) on p.88]{Davies1996}, for any $k\geq 1$, the $k$-th eigenvalue of the DtN operator on $\Omega$ is  equal to \begin{equation}\label{min_max_infiDefi}\sigma_k(\Omega):=\inf_{\substack{H\subset\ell_0(\delta\Omega),\\ \dim H=k}}\sup_{0\neq f\in H}\frac{\langle\Lambda (f),f\rangle_{\delta\Omega}}{\langle f,f\rangle_{\delta\Omega}}.\end{equation}
For $k=1,$ $\sigma_1(\Omega)$ is called the \emph{bottom spectrum} of $\Lambda.$ If the number of vertices in $\delta\Omega$ is finite, we will prove that $\sigma_1(\Omega)=0$ if and only if $\overline{\Omega}$ is recurrent, see Proposition~\ref{prop:recu}.

By exhaustion methods, in order to give Cheeger-type estimates for $\sigma_k(\Omega)$, $k\geq1$, we consider finite subsets of $\overline{\Omega}$. Let $W\subset\overline{\Omega}$ be a finite subset. Given any subset $A\subset W$, we denote by
$$\partial_{ W}A:=\partial A\cap E( W,\overline{ W})$$
the relative edge boundary of $A$ in $ W$. Then the Cheeger-type constants of $ W$ are defined as follows.

\begin{defi}\label{defi}Let $ W\subset\overline{\Omega}$ be a finite subset and $ W\cap\delta\Omega\neq\emptyset$. The Jammes-type Cheeger constant of $ W$ in $\overline{\Omega}$ is defined as
$$h_J( W):=\min_{\emptyset\neq A\subset W}\frac{\mu(\partial_{ W} A)}{d(A\cap\delta\Omega)}.$$
We set $h_J(W)=+\infty$ if $W\cap\delta\Omega=\emptyset$. Similarly, the classical Cheeger constant of $ W$ in $\overline{\Omega}$ is defined as
$${h}( W):=\min_{\emptyset\neq A\subset W}\frac{\mu(\partial_{ W} A)}{d(A)}.$$
\end{defi}
\begin{rem}
For any finite $ W\subset\overline{\Omega}$, it is easy to show that
\begin{equation}\label{eq:comp1}h_J( W)\geq {h}( W).\end{equation}
\end{rem}

The DtN operator (with Dirichlet boundary condition) on $ W,$ denoted by $\Lambda_ W$, is defined as
\[\aligned
\Lambda_{ W}: \R^{ W\cap\delta\Omega} & \to\R^{ W\cap\delta\Omega},
\\
f  &\mapsto \Lambda_{ W} (f):=\frac{\partial u^{ W}_f}{\partial n},
\endaligned\]
where $u^W_f$ is the solution of \eqref{existlemma}.

Let $\sigma_k( W)$, $1\leq k\leq\sharp(W\cap\delta\Omega)$, be the $k$-th eigenvalue of $\Lambda_{ W}$. Similar to \eqref{min_max_infiDefi}, $\sigma_k(W)$ can be characterized as
\begin{equation}\label{min_max_finite}\sigma_k(W):=\min_{\substack{H\subset\mathbb{R}^{W\cap\delta\Omega},\\ \dim H=k}}\max_{0\neq f\in H}\frac{\langle\Lambda_W (f),f\rangle_{W\cap\delta\Omega}}{\langle f,f\rangle_{W\cap\delta\Omega}}.\end{equation}

We obtain the following Jammes-type Cheeger estimate for $\sigma_1(W)$, the first nontrivial eigenvalue of $\Lambda_W$, which is an analog to \eqref{jammescheegerestimate} in the Riemannian case.

\begin{theorem}\label{finitejammeslowerestimate}
For any finite subset $ W\subset\overline{\Omega}$ with $ W\cap\delta\Omega\neq\emptyset$, we have
$$\frac{h( W)\cdot h_J( W)}{2}\leq\sigma_1( W)\leq h_J( W).$$
\end{theorem}
From Lemma \ref{eigenmono}, $\sigma_k( W)$ is non-increasing when $ W$ increases. Moreover, we obtain an approximation relation between ${\sigma}_k(\Omega)$ and $\{\sigma_k( W_i)\}^{\infty}_{i=1}$ for any $\mathcal{W}\uparrow\overline{\Omega}$.

\begin{prop}\label{eigenvalueapproximation}
For any $\mathcal{W}\uparrow\overline{\Omega}$ and $k\geq1,$
$$\lim_{\mathcal{W}\uparrow\overline{\Omega}}\sigma_k( W)={\sigma}_k(\Omega).$$
\end{prop}
As a corollary, we have the following estimate.
\begin{coro}\label{corodtn} For any $k\geq 1,$ $$\sigma_k(\Omega)\geq {\lambda}_k(\overline{\Omega}),$$ where ${\lambda}_k(\overline{\Omega})$ are eigenvalues of the normalized Laplacian on the graph $\overline{\Omega}=(\overline{\Omega},E(\Omega,\overline{\Omega}),\mu).$
\end{coro}

By Definition \ref{defi}, $h_J( W)$ and ${h}( W)$ are non-increasing when $ W$ increases. So by \eqref{exhaustionMono}, for any $\mathcal{W}\uparrow\overline{\Omega}$, the corresponding Cheeger constants for $\Omega$ can be defined as
$$h_J(\Omega):=\lim_{\mathcal{W}\uparrow\overline{\Omega}}h_J( W),$$ $${h}(\overline{\Omega}):=\lim_{\mathcal{W}\uparrow\overline{\Omega}}{h}( W).$$
Hence, by Theorem \ref{finitejammeslowerestimate} and Proposition \ref{eigenvalueapproximation}, we have the following estimates for the bottom spectrum ${\sigma}_1(\Omega)$.

\begin{theorem}\label{infinitejammeslowerestimate}
Let $\Omega\subset V$ be an infinite subset. We have
$$ \frac{{h}(\overline{\Omega})\cdot {h}_J(\Omega)}{2}\leq {\sigma}_1(\Omega)\leq{h}_J(\Omega).$$
\end{theorem}
\begin{rem}\label{rem:try1} Combining \eqref{eq:comp1} with Theorem \ref{infinitejammeslowerestimate}, we have
$${\sigma}_1(\Omega)\geq \frac12h^2(\overline{\Omega}).$$ This can be also derived from the Cheeger estimate on $\overline{\Omega}$ \cite{Fujiwara1996} and Corollary~\ref{corodtn} for $k=1.$ In particular, this yields that if $h(\overline{\Omega})>0,$ then ${\sigma}_1(\Omega)>0.$

\end{rem}

For any finite subset $ W\subset\overline{\Omega}$ with $ W\cap\delta\Omega\neq\emptyset$,
we denote by $\mathcal{A}( W)$ the collection of all nonempty subsets of $ W$ and $\mathcal{A}_k( W)$ the set of all disjoint $k$-tuples $(A_1,\cdots,A_k)$ such that $A_l\in\mathcal{A}( W)$, $\forall l\in[k]$, where $[k]:=\{1,\cdots,k\}$, $\forall k\in\mathbb{N}^+$. Following \cite{Miclo2017}, we define the higher order Cheeger-type constants for the DtN operator on $ W.$
\begin{defi}\label{hOrderCheegerConstant}
The $k$-th order Cheeger-Steklov constant for the DtN operator on $ W$ is defined as
$${h}_k( W):=\min_{(A_1,\cdots,A_k)\in\mathcal{A}_k( W)}\max_{l\in[k]}{h}_J(A_l){h}(A_l).$$
Similarly, the $k$-th order Jammes-type Cheeger constant for $ W$ is defined as
$${h}^k_{J}( W):=\min_{(A_1,\cdots,A_k)\in\mathcal{A}_k( W)}\max_{l\in[k]}\frac{\mu(\partial_{ W}A_l)}{d(A_l\cap\delta\Omega)}.$$
\end{defi}

Following \cite{Miclo2017}, as an intermediary step to obtain a Cheeger-type estimate for the higher order eigenvalues of the DtN operators, we prove a higher order Cheeger estimate for the Dirichlet problem on finite graphs, see Theorem \ref{hd} in the paper. Initiated by \cite[Proposition~3]{Miclo2017}, any $k$-th eigenvalue of the DtN operator on $W$ can be approximated by a sequence of the $k$-th eigenvalues of the Dirichlet Laplacians defined in \eqref{def:blowup} and \eqref{dtnapproximate}, see Proposition~\ref{approximation}.  Hence, combining Theorem \ref{hd} with Proposition \ref{approximation}, we obtain the following result.
\begin{theorem}\label{hdtn}
There exists an universal constant $c>0$ such that
$$\frac{c}{k^6}{h}_k( W)\leq \sigma_k( W)\leq{h}^k_{J}( W),$$
where $\sigma_k( W)$ is the $k$-th  eigenvalue of the DtN operator on $ W$.
\end{theorem}

By definition \ref{hOrderCheegerConstant}, ${h}_k( W)$ and ${h}^k_{J}( W)$ are non-increasing when $ W$ increases. Hence the corresponding Cheeger constants for $\Omega$ can be defined as
$${h}_k(\Omega):=\lim_{\mathcal{W}\uparrow\overline{\Omega}}{h}_k( W),$$
$${h}^k_{J}(\Omega):=\lim_{\mathcal{W}\uparrow\overline{\Omega}}{h}^k_{J}( W).$$

Finally, by exhaustion, the monotonicity of the higher order eigenvalues of the DtN operators, see Lemma \ref{eigenmono} in the paper, and the convergence of eigenvalues, see Proposition \ref{eigenvalueapproximation}, we have the following higher order Cheeger-type estimate for the DtN operator on infinite graphs.

\begin{theorem}\label{infinitehigherlowwerbound}
For any $k\in\mathbb{N}$, there exists a universal constant $c>0$ such that
$$\frac{c}{k^6}{h}_k(\Omega)\leq {\sigma}_k(\Omega)\leq{h}^k_{J}(\Omega).$$
\end{theorem}

\begin{rem}
Hassannezhad-Miclo \cite[Theorem~B]{Miclo2017} defined the DtN operator on a subset of a probability measure space $(M,\mathcal{M},\nu)$, endowed with a Markov kernel $P$ leaving $\nu$ invariant, and proved the higher order Cheeger estimate. Our result applies to general infinite graphs with possibly infinite total measure, which can be regarded as an extension of Hassannezhad-Miclo's result.
\end{rem}

The organization of the paper is as follows: In Section 2, we recall some facts on graphs. In Section 3, we study the spectra of the DtN operators on infinite subgraphs. In Section 4, we prove the Jammes-type Cheeger estimate for the bottom spectrum of the DtN operators. In Section 5, we obtain higher order Cheeger estimate for the Dirichlet problems. In Section 6, we prove higher order Cheeger estimate for the DtN operators on infinite subgraphs of graphs.

\section{Preliminaries}

Let $(X,\nu)$ be a discrete measure space, i.e. $X$ is a countable discrete space equipped with a Borel measure $\nu.$ For $p\in[1,\infty],$ the space of $\ell^p$ summable functions on $(X,\nu)$ is defined routinely: Given a function $f\in\R^X,$ for $p\in[1,\infty)$, we denote by
$$\|f\|_{\ell^p}=\left(\sum_{x\in X}|f(x)|^p\nu(x)\right)^{1/p}$$
the $\ell^p$ norm of $f.$ For $p=\infty,$
$$\|f\|_{\ell^\infty}=\sup_{x\in X}|f(x)|.$$
Let
$$\ell^p(X,\nu):=\{\left.f\in\R^X \right| \|f\|_{\ell^p}<\infty\}$$ be the space of $\ell^p$ summable functions on $(X,\nu).$ In our setting, these definitions apply to $(\Omega,d)$ and $(\delta\Omega,d)$ for $\Omega\subset V$ in a graph $(V,E,\mu).$ The case for $p=2$ is of particular interest, as we have the Hilbert spaces $\ell^2(\Omega,d)$ and $\ell^2(\delta\Omega,d)$ equipped with standard inner products
\[\aligned
\langle f,g\rangle_{\Omega} & =\sum_{x\in \Omega}f(x)g(x)d(x),\quad f,g\in \R^\Omega,
\\
\langle \varphi,\psi\rangle_{\delta\Omega} &=\sum_{x\in\delta\Omega}\varphi(x)\psi(x)d(x),\quad \varphi,\psi\in \R^{\delta\Omega}.
\endaligned\]

Given $ W\subset\overline{\Omega},$ an associated quadratic form is defined as
$$D_W(f,g)=\sum_{e=\{x,y\}\in E( W,\overline{ W})}\mu_{xy}(f(x)-f(y))(g(x)-g(y)), \quad f, g\in\R^{\overline{ W}}.$$
The Dirichlet energy of $f\in\R^{\overline{ W}}$ can be written as
\[D_W(f):=D_W(f,f).\]
For any $f\in\R^{\overline{\Omega}},$ the Laplacian of $f$ is defined as
$$\Delta f(x):=\frac{1}{d(x)}\sum_{y\in V: y\sim x}\mu_{xy}(f(y)-f(x)),\quad x\in \Omega.$$
For any $f\in \R^{\overline{\Omega}}$, the outward normal derivative of $f$ at $z\in\delta\Omega$ is defined as
\begin{equation}\label{neumannboundaryoperator}\frac{\partial f}{\partial n}(z):=\frac{1}{d(z)}\sum_{x\in\Omega: x\sim z}\mu_{zx}(f(z)-f(x)).\end{equation}
For any finite subset $W\subset\overline{\Omega}$, the Dirichlet eigenvalue problem on $W$ is defined as
\begin{equation}\label{normalizeDL}
\left\{
       \begin{array}{ll}
        \Delta f(x)=-\lambda f(x), &\quad x\in W, \\
        f(x)=0, &\quad x\in\delta W.
               \end{array}
     \right.
\end{equation}
We denote by $\lambda_{k,D}( W)$ the $k$-th eigenvalue of the above Dirichlet problem. For any $\mathcal{W}\uparrow\overline{\Omega}$, the $k$-th eigenvalue of the normalized Laplacian on the graph $\overline{\Omega}=(\overline{\Omega},E(\Omega,\overline{\Omega}),\mu)$ is defined as
\begin{equation}\label{infiKL}\lambda_k(\overline{\Omega}):=\lim_{i\rightarrow\infty}\lambda_{k,D}(W_i).\end{equation}

We recall the following well-known results for the Laplace operators, see e.g. \cite{Grigor'yan2009} for their proofs.
\begin{lemma}\label{Green's formula}
(Green's formula) For any finite subset $ W\subset\overline{\Omega}$ and any $f,g\in\R^{\overline{ W}}$, we have
\begin{equation}\label{eq:green formula}
\langle \Delta f,g\rangle_W=-D_{ W}(f,g)+\left\langle \frac{\partial f}{\partial n},g\right\rangle_{\delta W},
\end{equation}
where $\frac{\partial f}{\partial n}$ on $\delta W$ is defined similarly as \eqref{neumannboundaryoperator} with $\Omega$ replaced by $ W.$
\end{lemma}

\begin{lemma}\label{harmonic extension}
For any $f \in\R^{ W\cap\delta\Omega},$ there exists a unique function $u_{f}^{ W}\in\R^{\overline{ W}}$ satisfying \eqref{existlemma}.
\end{lemma}
We always denote by $u_{f}^{ W}$ the unique solution of \eqref{existlemma} with the boundary condition $f$ in this paper. By Green's formula, Lemma \ref{Green's formula}, we have the following lemma.
\begin{lemma}\label{fgreen}
For any finite subset $ W\subset\overline{\Omega}$ with $ W\cap\delta\Omega\neq\emptyset$ and  any $f,g\in\mathbb{R}^{ W\cap\delta\Omega},$ we have
$$D_{W}(u^{ W}_f,u^{ W}_g)=\left\langle\frac{\partial u^{ W}_f}{\partial n}, g\right\rangle_{ W\cap\delta\Omega}.$$
\end{lemma}

\section{DtN operators on infinite graphs}
Let $G=(V,E)$ be an infinite graph, $\Omega\subset V$ is an infinite subset. Let $\Lambda$ be the DtN operator on $\Omega$ defined in the introduction.

\begin{prop}\label{prop:selfadjoint} For any $f,g\in \ell_0(\delta\Omega),$
$$\langle\Lambda (f),g\rangle=\langle f, \Lambda (g)\rangle.$$
\end{prop}
\begin{proof} For sufficiently large $i$ such that $ W_i\supset \mathrm{supp}(f)\cup \mathrm{supp}(g),$ by Lemma~\ref{fgreen},
$$ \langle \frac{\partial u_f^{ W_i}}{\partial n}, g \rangle_{\delta\Omega}= \langle \frac{\partial u_g^{ W_i}}{\partial n}, f \rangle_{\delta\Omega}.$$
Since $f$ and $g$ are of finite support, only finitely many summands are involved in the above equation. By passing to the limit,
$$u_f^{ W_i}\to u_f, \quad u_g^{ W_i}\to u_g, \quad i\to\infty,$$ we prove the proposition.
\end{proof}

Let $f\in\ell_0(\delta\Omega)$, for any finite subset $ W\subset\overline{\Omega}$ with $ W\cap\delta\Omega\neq\emptyset$, set
$$\mathcal{L}(f, W):=\left\{\left.\phi\in\R^{\overline{\Omega}}\ \right|\ \mathrm{supp}(\phi)\subset W\quad\text{and}\quad\phi|_{ W\cap\delta\Omega}=f\right\}.$$
 We define the capacity of $W$ with boundary condition $f$ as
$$\mathrm{Cap}(f, W):=\inf_{\phi\in\mathcal{L}(f, W)}D_{ W}(\phi).$$
Similarly, the capacity of $\overline{\Omega}$ with boundary condition $f$ is defined as
$$\mathrm{Cap}(f):=\inf_{\phi\in \ell_0(\overline{\Omega}), \phi|_{\delta\Omega}=f}D_{\Omega}(\phi).$$

For any $\mathcal{W}\uparrow\overline{\Omega},$ since $f\in\ell_0(\delta\Omega)$, there exists $M\in\N^+$, such that $\mathrm{supp}(f)\subset W_i\cap\delta\Omega$, $\forall i>M$. Note that $\mathcal{L}(f, W_{i_1})\subset\mathcal{L}(f, W_{i_2})$, for any $i_1,i_2>M$ and $i_1<i_2$. Hence by definition, $\mathrm{Cap}(f, W_i)$ is non-increasing when $i> M$. One can verify that
\begin{equation}\label{capacity}\mathrm{Cap}(f)=\lim_{i\to \infty}\mathrm{Cap}(f, W_i).\end{equation}

\begin{lemma}\label{51}
$$\mathrm{Cap}(f, W)=D_{ W}(u^{ W}_f).$$
\end{lemma}
\begin{proof}
\begin{eqnarray*}
&&\mathrm{Cap}(f, W)=\inf_{\phi\in\mathcal{L}(f, W)}D_{ W}(\phi)\nonumber\\
&=&\inf_{\phi\in\mathcal{L}(f, W)}D_{ W\cap\Omega}(\phi)+\sum_{x\in W\cap\delta\Omega}\sum_{y\in\Omega\setminus W}\mu_{xy}(f(x)-0)^2\nonumber\\
&=&D_{ W\cap\Omega}(u^{ W}_{f})+\sum_{x\in W\cap\delta\Omega}\sum_{y\in\Omega\setminus W}\mu_{xy}f^2(x)\nonumber\\
&=&D_{ W}(u^{ W}_{f}).
\end{eqnarray*}
The second last equality follows from the fact that the harmonic function $u^W_f$ minimizes the Dirichlet energy among functions with the same boundary condition.
\end{proof}

By Lemma \ref{51} and \ref{fgreen}, we have
\begin{equation}\label{52}
\mathrm{Cap}(f, W)=D_{ W}(u^{ W}_f)=\left\langle\frac{\partial u^{ W}_f}{\partial n},f\right\rangle_{ W\cap\delta\Omega}.
\end{equation}

\begin{prop}\label{dtndirichletenergy}
For any $f\in\ell_0(\delta\Omega)$, we have
$$\mathrm{Cap}(f)=D_{\Omega}(u_f)=\left\langle \Lambda (f),f\right\rangle.$$
\end{prop}

\begin{proof}
For any $\mathcal{W}\uparrow\overline{\Omega},$ since $f\in\ell_0(\delta\Omega)$, there exists $M\in\N$, such that $\mathrm{supp}(f)\subset W_i\cap\delta\Omega$, $\forall i>M$. By \eqref{52},
$$\mathrm{Cap}(f, W_i)=D_{ W_i}(u_f^{ W_i})= \langle\frac{\partial u_f^{ W_i}}{\partial n},f \rangle_{ W_i\cap\delta\Omega}.$$
Letting $i\to \infty$, by \eqref{capacity} we have
\begin{eqnarray}\label{3}
\mathrm{Cap}(f)&=&\lim_{i\to \infty}\mathrm{Cap}(f, W_i)\nonumber\\
&=&\lim_{i\to \infty}\langle\frac{\partial u_f^{ W_i}}{\partial n},f\rangle_{ W_i\cap\delta\Omega}=\langle\frac{\partial u_f}{\partial n}, f\rangle_{\delta\Omega},\end{eqnarray}
where the last equality follows from that $u^{ W_i}_f$ converges to $u_f$ pointwise and $f\in\ell_0(\delta\Omega)$.
By Fatou's lemma,
\begin{eqnarray}\label{4}D_{\Omega}(u_f)\leq\liminf_{i\to \infty}D_{ W_i}(u_f^{ W_i})=\liminf_{i\to \infty}\mathrm{Cap}(f, W_i)=\mathrm{Cap}(f).\end{eqnarray}
For any $i>M$,
\begin{eqnarray}\label{1}D_{ W_i}(u_f-u_f^{ W_i})=D_{ W_i}(u_f)+D_{ W_i}(u_f^{ W_i})-2D_{ W_i}(u_f,u_f^{ W_i}).\end{eqnarray}
By Green's formula, Lemma \ref{Green's formula}, the last term in \eqref{1} can be written as
\begin{eqnarray}\label{2}D_{ W_i}(u_f,u_f^{ W_i})&=&\langle \frac{\partial u_f}{\partial n},u_f^{ W_i}\rangle_{ W_i\cap\delta\Omega}\nonumber=\langle \frac{\partial u_f}{\partial n},f\rangle_{\delta\Omega}=\mathrm{Cap}(f),
\end{eqnarray}
where the last equality follows from \eqref{3}. Therefore, \eqref{1} implies that
$$0\leq D_{ W_i}(u_f-u_f^{ W_i})=D_{ W_i}(u_f)+\mathrm{Cap}(f, W_i)-2\mathrm{Cap}(f),$$
i.e.
$$D_{ W_i}(u_f)\geq 2\mathrm{Cap}(f)-\mathrm{Cap}(f, W_i),$$
whence by letting $i\to \infty$, we get
\begin{equation}\label{geqCap}D_{\Omega}(u_f)=\lim_{i\rightarrow\infty}D_{W_i}(u_f)\geq 2\mathrm{Cap}(f)-\lim_{i\rightarrow\infty}\mathrm{Cap}(f,W_i)=\mathrm{Cap}(f).\end{equation}
Combining \eqref{4} with \eqref{geqCap}, we have
\begin{eqnarray}\label{5}D_{\Omega}(u_f)=\mathrm{Cap}(f).\end{eqnarray}
Then the proposition follows from \eqref{3} and \eqref{5}.
\end{proof}

The proof of the above proposition yields the following corollary.
\begin{coro}\label{coro524}
For any $\mathcal{W}\uparrow\overline{\Omega}$ and $f\in\ell_0(\delta\Omega),$ we have
$$D_{\Omega}(u_f)=\lim_{\mathcal{W}\uparrow\overline{\Omega}}D_{ W}(u^{ W}_{f}).$$
\end{coro}
Combining \eqref{min_max_infiDefi} with Proposition \ref{dtndirichletenergy}, we have the following corollary.
\begin{coro}\label{energyConverge}
\begin{eqnarray*}\sigma_k(\Omega):&=&\inf_{\substack{H\subset\ell_0(\delta\Omega),\\ \dim H=k}}\sup_{0\neq f\in H}\frac{\langle\Lambda (f),f\rangle_{\delta\Omega}}{\langle f,f\rangle_{\delta\Omega}}\nonumber\\
&=&\inf_{\substack{H\subset\ell_0(\delta\Omega),\\ \dim H=k}}\sup_{0\neq f\in H}\frac{D_{\Omega}(u_f)}{\langle f,f\rangle_{\delta\Omega}}\nonumber\\
&=&\inf_{\substack{H\subset\ell_0(\delta\Omega),\\ \dim H=k}}\sup_{0\neq f\in H}\frac{\mathrm{Cap}(f)}{\langle f,f\rangle_{\delta\Omega}}.
\end{eqnarray*}
\end{coro}

In the next lemma, we show that $\Lambda$ is a bounded operator on $\ell^2(\delta\Omega).$
\begin{lemma}\label{welldefined}
For any $f\in\ell_0(\delta\Omega),$ we have
$$\parallel\Lambda(f)\parallel_{\ell^2(\delta\Omega)}\leq\parallel f\parallel_{\ell^2(\delta\Omega)},$$
i.e. $\Lambda$ is a bounded linear operator on $\ell_0(\delta\Omega)$.
\end{lemma}

\begin{proof}
\begin{eqnarray*}
\left\|\frac{\partial u_{f}}{\partial n}\right\|_{\ell^2(\delta\Omega)}^2
&=&\sum_{x\in\delta\Omega}\left|\frac{1}{d(x)}\sum_{y\in\Omega}\mu_{xy}(u_{f}(x)-u_{f}(y))\right|^2d(x)\nonumber\\
&\leq&\sum_{x\in\delta\Omega}\sum_{y\in\Omega}\mu_{xy}(u_{f}(x)-u_{f}(y))^2\nonumber\\
&\leq&D_{\Omega}(u_{f}).
\end{eqnarray*}
For any $f\in\ell_0(\delta\Omega)$ and any finite subset $W\subset\overline{\Omega}$, we denote by
\[
\overline{f}_W(x)=\left\{
       \begin{array}{ll}
       f(x) ,& \quad x\in W\cap\delta\Omega \\
       0 ,& \quad \mathrm{otherwise}
               \end{array}
     \right.
\]
the zero extension of $f|_{W\cap\delta\Omega}$.
By Corollary \ref{coro524},
\begin{eqnarray*}
&&D_{\Omega}(u_{f})=\lim_{\mathcal{W}\uparrow\overline{\Omega}}D_{ W}(u^{ W}_{f})\nonumber\\
&=&\lim_{\mathcal{W}\uparrow\overline{\Omega}}\left(\sum_{e=\{x,y\}\in E( W\cap\Omega,\overline{ W\cap\Omega})}+\sum_{x\in W\cap\delta\Omega}\sum_{y\in\Omega\setminus W}\right)
\mu_{xy}(u^{ W}_f(x)-u^{ W}_f(y))^2\nonumber\\
&\leq&\lim_{\mathcal{W}\uparrow\overline{\Omega}}\left(D_{W\cap\Omega}(\overline{f}_W)+\sum_{x\in W\cap\delta\Omega}\sum_{y\in\Omega\setminus W}
\mu_{xy}f^2(x)\right)\nonumber\\
&=&\lim_{\mathcal{W}\uparrow\overline{\Omega}}\left(\sum_{x\in W\cap\delta\Omega}f^2(x)\sum_{y\in W\cap\Omega}\mu_{xy}+\sum_{x\in W\cap\delta\Omega}\sum_{y\in\Omega\setminus W}
\mu_{xy}f^2(x)\right)\nonumber\\
&=&\lim_{\mathcal{W}\uparrow\overline{\Omega}}\|f\|_{\ell^2( W\cap\delta\Omega)}=\|f\|_{\ell^2(\delta\Omega)},
\end{eqnarray*}
where the inequality above follows from the fact that the harmonic function $u^W_f$ minimizes the Dirichlet energy among functions with the same boundary condition, in particular compared with $\overline{f}_W$.

Hence, we have
$$\left\|\frac{\partial u_{f}}{\partial n}\right\|_{\ell^2(\delta\Omega)}^2\leq\|f\|_{\ell^2(\delta\Omega)}^2.$$

\end{proof}

By the boundedness of $\Lambda$ on $\ell_0(\delta\Omega)$ and the density of $\ell_0(\delta\Omega)$ in $\ell^2(\delta\Omega),$ $\Lambda$ can be uniquely extended to $\ell^2(\delta\Omega).$

For any finite subset $ W\subset\overline{\Omega}$ with $ W\cap\delta\Omega\neq\emptyset$, we denote by $\sigma_k( W)$ the $k$-th eigenvalue of the DtN operator $\Lambda_{ W}$. By \eqref{min_max_finite} and \eqref{52},  $\sigma_k( W)$, $1\leq k\leq\sharp(W\cap\delta\Omega)$, can be characterized as
\begin{eqnarray}\sigma_k(W):&=&\min_{\substack{H\subset\mathbb{R}^{W\cap\delta\Omega},\\ \dim H=k}}\max_{0\neq f\in H}\frac{\langle\Lambda_W (f),f\rangle_{W\cap\delta\Omega}}{\langle f,f\rangle_{\delta\Omega}}\nonumber\\
&=&\min_{H\subset\mathbb{R}^{ W\cap\delta\Omega},\dim H=k}\max_{0\neq f\in H}\frac{\mathrm{Cap}(f, W)}{\langle f,f\rangle_{ W\cap\delta\Omega}}\nonumber\\
&=&\min_{\substack{H\subset\mathbb{R}^{W\cap\delta\Omega},\\ \dim H=k}}\max_{0\neq f\in H}\frac{D_W(u^W_f)}{\langle f,f\rangle_{\delta\Omega}}\label{min_max_finiDefi}.\end{eqnarray}
In order to give Cheeger estimates for infinite graphs, we need the following monotonicity result.

\begin{lemma}\label{eigenmono}
For any $\mathcal{W}\uparrow\overline{\Omega}$,
$$\sigma_k( W_i)\geq\sigma_k( W_{i+1}),\quad\forall i=1,2,\cdots.$$
where $1\leq k\leq\sharp(W_i\cap\delta\Omega)$.
\end{lemma}
\begin{proof}
For any $i$, by \eqref{min_max_finiDefi}, one can choose $H\subset\R^{ W_i\cap\delta\Omega}$, $\dim H=k$, such that
$$\sigma_k( W_i)=\max_{0\neq f\in H}\frac{\mathrm{Cap}(f,W_i)}{\langle f,f\rangle_{ W_i\cap\delta\Omega}}.$$
Then we have
\begin{eqnarray*}
\sigma_k( W_i)&=&\max_{0\neq f\in H}\frac{\mathrm{Cap}(f, W_i)}{\langle f,f\rangle_{ W_i\cap\delta\Omega}}\geq\max_{0\neq f\in H}\frac{\mathrm{Cap}(f, W_{i+1})}{\langle f,f\rangle_{ W_{i+1}\cap\delta\Omega}}\nonumber\\
&\geq&\min_{\substack{H'\subset\mathbb{R}^{W_{i+1}\cap\delta\Omega}\\\dim H'=k}}\max_{0\neq f\in H'}\frac{\mathrm{Cap}(f, W_{i+1})}{\langle f,f\rangle_{ W_{i+1}\cap\delta\Omega}}\nonumber\\
&=&\sigma_{k}( W_{i+1}).
\end{eqnarray*}
\end{proof}

Now we are ready to prove the approximation result in Proposition \ref{eigenvalueapproximation}.

\begin{proof}[Proof of Proposition \ref{eigenvalueapproximation}]
For any $i\in\mathbb{N}$, by \eqref{min_max_finiDefi}, choose $H\subset\R^{ W_i\cap\delta\Omega}$, $\dim H=k$, such that
$$\sigma_k( W_i)=\max_{0\neq f\in H}\frac{\mathrm{Cap}(f,W_i)}{\langle f,f\rangle_{ W_i\cap\delta\Omega}}.$$
Then by the monotonicity of $\mathrm{Cap}(f, W_i)$ and the definition of $\mathrm{Cap}(f)$,
\begin{eqnarray*}
\sigma_k( W_i)&=&\max_{0\neq f\in H}\frac{\mathrm{Cap}(f, W_i)}{\langle f,f\rangle_{ W_i\cap\delta\Omega}}\geq\max_{0\neq f\in H}\frac{\mathrm{Cap}(f)}{\langle f,f\rangle_{ W_i\cap\delta\Omega}}\nonumber\\
&=&\max_{0\neq f\in H}\frac{D_{\Omega}(u_f)}{\langle f,f\rangle_{\delta\Omega}}\geq\inf_{\substack{H'\subset\ell_0(\delta\Omega)\\\dim H'=k}}\sup_{0\neq f\in H'}\frac{D_{\Omega}(u_f)}{\langle f,f\rangle_{\delta\Omega}}\nonumber\\
&=&{\sigma}_k(\Omega),
\end{eqnarray*}
The second last equality follows from Proposition \ref{dtndirichletenergy}. Hence
$$\lim_{i\to \infty}\sigma_k( W_i)\geq{\sigma}_k(\Omega).$$
On the other hand, by Corollary \ref{energyConverge}, for any $\epsilon>0$, there exists $H\subset\ell_0(\delta\Omega)$, $\dim H=k$, such that
$${\sigma}_k(\Omega)\leq \sup_{f\in H}\frac{D_{\Omega}(u_f)}{\langle f,f\rangle_{\delta\Omega}}=\max_{f\in H}\frac{D_{\Omega}(u_f)}{\langle f,f\rangle_{\delta\Omega}}<{\sigma}_k(\Omega)+\epsilon.$$
For the above $H$, there exists $K\in\mathbb{N}^+$, such that
$$\mathrm{supp}(g)\subset W_i\cap\delta\Omega, \quad\forall i>K,\, \forall g\in H.$$
Hence
$$\sigma_k(W_i)\leq\max_{f\in H}\frac{D_{\Omega}(u_f)}{\langle f,f\rangle_{\delta\Omega}}<\sigma_k(\Omega)+\epsilon.$$
Letting $i\to \infty$ and $\epsilon\rightarrow 0$, we have
$$\lim_{i\to \infty}\sigma_k( W_i)\leq{\sigma}_k(\Omega).$$
Hence the proposition follows.
\end{proof}

Now we are ready to prove Corollary~\ref{corodtn}.
\begin{proof}[Proof of Corollary~\ref{corodtn}]
By Proposition \ref{eigenvalueapproximation} and \eqref{infiKL}, it suffices to prove that for any $\mathcal{W}\uparrow\overline{\Omega}$,
$$\sigma_k( W)\geq\lambda_{k,D}( W),$$
where $\lambda_{k,D}(W)$ is the $k$-th eigenvalue of the Dirichlet problem \eqref{normalizeDL}. Given any subset $H\subset\mathbb{R}^{W\cap\delta\Omega}$,
we denote by
$$\widetilde{H}:=\mathrm{span}\left\{u^W_f|f\in H\right\}.$$
By definition,
\begin{eqnarray*}
\sigma_k( W)&=&\min_{\substack{H\subset\mathbb{R}^{ W\cap\delta\Omega}\\\dim H=k}}\max_{0\neq f\in H}\frac{\sum_{e=\{x,y\}\in E( W,\overline{ W})}\mu_{xy}(u^{ W}_f(x)-u^{ W}_f(y))^2}{\sum_{x\in W\cap\delta\Omega}f^2(x)d(x)}\nonumber\\
&=&\min_{\substack{\widetilde{H}\subset\mathbb{R}^{W}\\\dim \widetilde{H}=k}}\max_{0\neq f\in \widetilde{H}}\frac{\sum_{e=\{x,y\}\in E( W,\overline{ W})}\mu_{xy}(f(x)-f(y))^2}{\sum_{x\in W\cap\delta\Omega }f^2(x)d(x)}\nonumber\\
&\geq&\min_{\substack{\widetilde{H}\subset\mathbb{R}^{W}\\\dim \widetilde{H}=k}}\max_{0\neq f\in \widetilde{H}}\frac{\sum_{e=\{x,y\}\in E( W,\overline{ W})}\mu_{xy}(f(x)-f(y))^2}{\sum_{x\in W }f^2(x)d(x)}\nonumber\\
&\geq&\min_{\substack{H'\subset\mathbb{R}^{ W}\\\dim(H')=k}}\max_{0\neq f\in H'}\frac{\sum_{e=\{x,y\}\in E( W,\overline{ W})}\mu_{xy}(f(x)-f(y))^2}{\sum_{x\in W }f^2(x)d(x)}\nonumber\\
&=&\lambda_{k,D}( W).
\end{eqnarray*}
Hence the proposition follows.
\end{proof}

\begin{rem}
Let $\Omega$ be a finite subset and $k=2.$ Corollary \ref{corodtn} and the fact $\lambda_2(\Omega)\leq 2$ imply that
$$\sigma_2(\Omega)\geq\lambda_2(\overline{\Omega})\geq \frac{1}{8}\lambda_2^2(\overline{\Omega}),$$ which is the consequence of Corollary~{1.1} in \cite{BoboYanZuoqin2017}.

\end{rem}

\section{Jammes-type Cheeger estimate for the bottom spectrum}
Let $(V, E,\mu)$ be an infinite graph, $\Omega\subset V$ be an infinite subgraph and $ W\subset\overline{\Omega}$ be a finite subset with $ W\cap\delta\Omega\neq\emptyset$. Let $0\neq f\in\mathbb{R}^{ W\cap\delta\Omega}$ be the first eigenfunction associated to the first eigenvalue $\sigma_1( W)$. For convenience, we write $f$ for $u_f^W$ in the following. Without loss of generality, we may assume that $f$ is nonnegative, since $D_{ W}(|f|)\leq D_{ W}(f).$  By \eqref{52}, we have
\begin{eqnarray}\label{infinitelemma1}
\sigma_1( W)=\frac{\sum_{e=\{x,y\}\in E( W,\overline{ W})}\mu_{xy}(f(y)-f(x))^2}{\sum_{x\in W\cap\delta\Omega}f^2(x)d(x)}.
\end{eqnarray}

Multiplying both the numerator and denominator of the fraction at the right hand side of \eqref{infinitelemma1} by $\sum_{x\in W}f^2(x)d(x)$ and setting
$$\frac{M}{N}:=\frac{\sum_{x\in W}f^2(x)d(x)\cdot\sum_{e=\{x,y\}\in E( W,\overline{ W})}\mu_{xy}(f(y)-f(x))^2}{\sum_{x\in W}f^2(x)d(x)\cdot\sum_{x\in W\cap\delta\Omega}f^2(x)d(x)},$$
we have
$$\sigma_1( W)=\frac{M}{N}.$$
We need the following lemmas to prove Theorem \ref{finitejammeslowerestimate}.

\begin{lemma}\label{lemma62}
$$M\geq\frac{1}{2}\left(\sum_{e=\{x,y\}\in E( W,\overline{ W})}\mu_{xy}|f^2(x)-f^2(y)|\right)^2.$$
\end{lemma}
\begin{proof}
Note that
\begin{eqnarray*}
&&\sum_{x\in W}f^2(x)d(x)=\left(\sum_{x\in W}\sum_{y\in W}+\sum_{x\in W}\sum_{y\in\overline{\Omega}\setminus W}\right)\mu_{xy}f^2(x)\nonumber\\
&=&\frac12\sum_{x,y\in W}\mu_{xy}(f^2(x)+f^2(y))+\sum_{x\in W}\sum_{y\in\overline{\Omega}\setminus W}\mu_{xy}f^2(x)\nonumber\\
&=&\sum_{e=\{x,y\}\in E( W, W)}\mu_{xy}(f^2(x)+f^2(y))+\sum_{x\in W}\sum_{y\in\overline{\Omega}\setminus W}\mu_{xy}(f^2(x)+f^2(y))\nonumber\\
&=&\sum_{e=\{x,y\}\in E( W,\overline{ W})}\mu_{xy}(f^2(x)+f^2(y)).
\end{eqnarray*}
Hence
\begin{eqnarray*}
M&=&\sum_{e=\{x,y\}\in E( W,\overline{ W})}\mu_{xy}(f^2(x)+f^2(y))\cdot\sum_{e=\{x,y\}\in E( W,\overline{ W})}\mu_{xy}(f(y)-f(x))^{2}\nonumber\\
&\geq&\frac12\sum_{e=\{x,y\}\in E( W,\overline{ W})}\mu_{xy}(f(x)+f(y))^2\cdot\sum_{e=\{x,y\}\in E( W,\overline{ W})}\mu_{xy}(f(y)-f(x))^{2}\nonumber\\
&\geq&\frac12\left(\sum_{e=\{x,y\}\in E( W,\overline{ W})}\mu_{xy}|f^2(x)-f^2(y)|\right)^2.
\end{eqnarray*}
The last inequality follows from H\"{o}lder's inequality.
\end{proof}

Set $t_0:=\max_{x\in W}\{f(x)\}$. For any $t>0$, set
$$S_t:=f^{-1}([\sqrt{t},+\infty))=\{x\in W|f^2(x)\geq t\}.$$
By the maximum principle, the maximizer of $f$ can't be achieved in $ W$, hence $S_t\cap\delta\Omega\neq\emptyset$ for any $t\in(0,t_0]$.

\begin{lemma}\label{lemma63}
$$\int_0^{\infty}\mu(\partial_{ W}S_t)dt=\sum_{e=\{x,y\}\in E( W,\overline{ W})}\mu_{xy}|f^2(x)-f^2(y)|.$$
\end{lemma}
\begin{proof}
For any interval $(a,b]$, we denote by $\chi_{(a,b]}$ the characteristic function on $(a,b]$, i.e.
\[
\chi_{(a,b]}(x)=\left\{
       \begin{array}{ll}
        0,& x\notin (a,b], \\
        1,&x\in (a,b].
               \end{array}
     \right.
\]
Then
\begin{eqnarray*}
\int_0^{\infty}\mu(\partial_{ W}S_t)dt&=&\int_0^{\infty}\sum_{e=\{x,y\}\in E( W,\overline{ W}),f^2(y)<t\leq f^2(x)}\mu_{xy}dt\nonumber\\
&=&\int_0^{\infty}\sum_{e=\{x,y\}\in E( W,\overline{ W})}\mu_{xy}\chi_{(f^2(y),f^2(x)]}(t)dt\nonumber\\
&=&\sum_{e=\{x,y\}\in E( W,\overline{ W})}\mu_{xy}|f^2(x)-f^2(y)|.
\end{eqnarray*}
\end{proof}
\begin{rem}
Lemma \ref{lemma63} is indeed the Coarea formula in the discrete setting, see \cite[Lemma 3.3]{Grigor'yan2009}.
\end{rem}

\begin{lemma}\label{lemma64}
$$\int_0^{\infty}d(S_t)dt=\sum_{x\in W}f^2(x)d(x).$$
$$\int_0^{\infty}d(S_t\cap\delta\Omega)dt=\sum_{x\in W\cap\delta\Omega}f^2(x)d(x).$$
\end{lemma}
\begin{proof}
By calculation,
\begin{eqnarray*}
\int_0^{\infty}d(S_t)dt=\int_0^{\infty}\sum_{x\in S_t}d(x)dt=\int_0^{\infty}\sum_{x\in W}d(x)\chi_{(0,f^2(x)]}(t)dt
=\sum_{x\in W}f^2(x)d(x).
\end{eqnarray*}
Similarly
\begin{eqnarray*}
&&\int_0^{\infty}d(S_t\cap\delta\Omega)dt=\int_0^{\infty}\sum_{x\in S_t\cap\delta\Omega}d(x)dt\nonumber\\
&=&\int_0^{\infty}\sum_{x\in W\cap\delta\Omega}d(x)\chi_{(0,f^2(x)]}(t)dt
=\sum_{x\in W\cap\delta\Omega}f^2(x)d(x).
\end{eqnarray*}
Hence we complete the proof.
\end{proof}

Now we are ready to prove Theorem~\ref{finitejammeslowerestimate}.
\begin{proof}[Proof of Theorem~\ref{finitejammeslowerestimate}]

For the upper bound estimate, choose $A\subset W$ that achieves ${h}_J( W)$, i.e.
$${h}_J( W)=\frac{\mu(\partial_{ W}A)}{d(A\cap\delta\Omega)}.$$
Set $g=\chi_A\in\R^{\overline{ W}},$ i.e.
\[
g(x)=\left\{
       \begin{array}{ll}
        1,& x\in A, \\
        0,&x\in \overline{ W}\setminus A.
               \end{array}
     \right.
\]

Then we have
\begin{eqnarray*}
\sigma_1( W)\leq\frac{D_{ W}(g)}{\langle g,g\rangle_{ W\cap\delta\Omega}}=\frac{\mu(\partial_{ W}A)}{d(A\cap\delta\Omega)}={h}_J( W).
\end{eqnarray*}

For the lower bound estimate, combining Lemma \ref{infinitelemma1} to Lemma \ref{lemma64}, we have
\begin{eqnarray*}
\sigma_1( W)&\geq&\frac{1}{2}\frac{\int_{0}^{\infty}\mu(\partial_{ W}S_t)dt\cdot\int_{0}^{\infty}\mu(\partial_{ W}S_t)dt}{\sum_{x\in W}f^2(x)d(x)\cdot\sum_{x\in
 W\cap\delta\Omega}f^2(x)d(x)}\nonumber\\
&\geq&\frac{1}{2}\frac{\int_{0}^{\infty}{h}( W)d(S_t)dt\cdot\int_{0}^{\infty}{h}_J( W)
d(S_t\cap\delta_S W)dt}{\sum_{x\in W}f^2(x)d(x)\cdot\sum_{x\in
 W\cap\delta\Omega}f^2(x)d(x)}\nonumber\\
&=&\frac{{h}( W){h}_J( W)}{2}\frac{\int_{0}^{\infty}d(S_t)dt\cdot\int_{0}^{\infty}d( S_t\cap\delta\Omega)dt}{\sum_{x\in W}f^2(x)d(x)\cdot\sum_{x\in
 W\cap\delta\Omega}f^2(x)d(x)}\nonumber\\
&=&\frac{{h}( W){h}_J( W)}{2}.
\end{eqnarray*}

The theorem follows from the above estimates.
\end{proof}

Finally,
we are ready to prove Theorem~\ref{infinitejammeslowerestimate}.

\begin{proof}[Proof of Theorem~\ref{infinitejammeslowerestimate}] For the upper bound estimate, by Proposition~\ref{eigenvalueapproximation} and Theorem~\ref{finitejammeslowerestimate},
$${\sigma}_1(\Omega)=\lim_{\mathcal{W}\uparrow\overline{\Omega}}\sigma_1( W)\leq \lim_{\mathcal{W}\uparrow\overline{\Omega}}{h}_J( W)={h}_J(\Omega).$$

Similarly, we have the lower bound estimate
\begin{eqnarray*}
{\sigma}_1(\Omega)=\lim_{\mathcal{W}\uparrow\overline{\Omega}}\sigma_1( W)
\geq\lim_{\mathcal{W}\uparrow\overline{\Omega}}\frac{{h}( W)\cdot {h}_J( W)}{2}
=\frac{{h}(\overline{\Omega})\cdot {h}_J(\Omega)}{2}.
\end{eqnarray*}

Hence we complete the proof of the theorem.
\end{proof}

At the end of this section, we give a necessary and sufficient condition for the positivity of ${\sigma}_1(\Omega)$ for an infinite subgraph $\Omega$ with finite vertex boundary. \begin{defi}
Let $(V,E,\mu)$ be an infinite graph. For any finite subset $F\subset V$, we define
$$\mathrm{Cap}_{F}(V):=\inf_{\substack{0\leq\phi\leq1,\phi|_{F}=1\\\mathrm{supp}(\phi)\subset V}}D_{V}(\phi).$$
\end{defi}
In order to obtain the sufficient condition for the positivity of $\sigma_1(\Omega)$, we need the following criterion for an infinite graph to be recurrent, see e.g. \cite[Theorem 2.12]{Woess2000}.
\begin{lemma}
Infinite graph $(\overline{\Omega},E(\Omega,\overline{\Omega}),\mu)$ is recurrent if and only if $\mathrm{Cap}_F(\overline{\Omega})=0$ for any finite subset $F\subset\overline{\Omega}$.
\end{lemma}
Then we have the following result.
\begin{prop}\label{prop:recu}
If $\Omega$ is an infinite subgraph of $G=(V,E)$ and $\sharp\delta\Omega<\infty$. Then ${\sigma}_1(\Omega)=0$ if and only if $\overline{\Omega}=(\overline{\Omega},E(\Omega,\overline{\Omega}))$ is recurrent.
\end{prop}
\begin{proof}
For the ``if " part: If $\overline{\Omega}$ is recurrent,  then $\mathrm{Cap}_{F}(\overline{\Omega})=0$ for any finite subset $F\subset\overline{\Omega}$. Choose $F=\delta\Omega$, we have
$$\mathrm{Cap}_{\delta\Omega}(\overline{\Omega})=\inf_{\substack{0\leq\phi\leq1,\phi|_{\delta\Omega}=1\\
\mathrm{supp}(\phi)\subset\overline{\Omega}}}D_{\Omega}(\phi)=0.$$
Hence, ${\sigma}_1(\Omega)=0$.

For the ``only if " part: If ${\sigma}_1(\Omega)=0$, then there exist $\{f_i\}^{\infty}_{i=0}\subset\mathbb{R}^{\delta\Omega}$, $\|f_i\|_{\ell^2(\delta\Omega)}=1$, such that
$$D_{\Omega}(u_{f_i})\rightarrow 0, \ i\to \infty.$$
Hence there exist a subsequence $\{f_{i_{j}}\}^{\infty}_{j=1}$ and  $f_{\infty}$ such that
$$f_{i_{j}}\rightarrow f_{\infty},\ j\to\infty\quad\mathrm{and}\quad \|f_{\infty}\|_{\ell^2(\delta\Omega)}=1.$$
Moreover, by the exhaustion method and the maximum principle on finite set, one can show that $$u_{f_{i_j}}\rightarrow u_{f_{\infty}},\ j\to\infty.$$
Then by the lower semi-continuity,
$$D_{\Omega}(u_{f_{\infty}})\leq\lim_{j\rightarrow\infty}D_{\Omega}(u_{f_{i_j}})=0.$$
Hence $u_{f_{\infty}}=\mathrm{const}$. This implies that $\mathrm{Cap}_{\delta\Omega}(\overline{\Omega})=0$ and $\overline{\Omega}$ is recurrent.

\end{proof}

We give an example with positive bottom spectrum for the DtN operator.
\begin{example}
Let $\Omega$ be a part of the homogenous tree with degree three and $\sharp\delta\Omega=1$, see Figure 1. By calculation, $h_J(\Omega)=1$, $h(\Omega)=\frac{1}{3}$, and ${\sigma}_1(\Omega)=\frac{1}{2}.$ See the related discussions in Remark~\ref{rem:try1}.
\end{example}

\begin{figure}[!h]
\includegraphics[height=5cm,width=9cm]{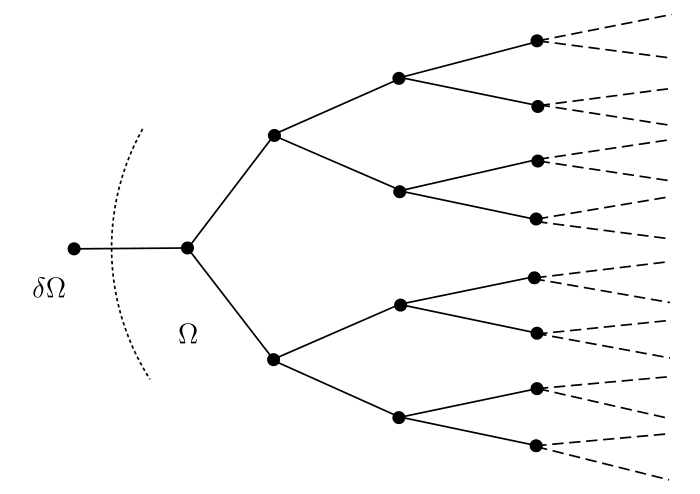}
\caption{}\label{figure1}
\end{figure}

\section{Higher order Cheeger estimates for the Dirichlet eigenvalue problems}
Higher order Cheeger estimates for the Laplace operator on finite graphs (without boundary condition) have been proved in \cite{Trevsan2014}. In this section, we prove higher order Cheeger estimates for the eigenvalues of the Dirichlet Laplacian problem on finite graphs.

Let $V$ be a countable set of vertices, and $$\mu:V\times V\rightarrow[0,\infty),\quad \{x,y\}\mapsto\mu_{xy}=\mu_{yx}$$ be a symmetric weight function. Instead of the vertex measure $d(\cdot)$ defined in the introduction, we introduce a general measure $\nu(\cdot)$ on $V,$
\[\aligned
\nu: V & \to (0,\infty),\\
x &\mapsto \nu(x).
\endaligned\]
This induces a (general) weighted graph structure $G=(V,\mu,\nu)$. For any $f\in \R^V$, the Laplacian on $(V,\mu,\nu)$ is defined as
$$\Delta^{\mu}_{\nu}(f)(x):=\frac{1}{\nu(x)}\sum_{y\in V}\mu_{xy}(f(y)-f(x)),\quad\forall x\in V.$$

Let $ W\subset V$ be a finite subset. Consider the following Dirichlet eigenvalue problem
\begin{equation}\label{Dirichlet}
\left\{
       \begin{array}{ll}
        \Delta^{\mu}_{\nu}(f)(x)=-\lambda f(x), &\quad x\in W, \\
        f(x)=0, &\quad x\in\delta W.
               \end{array}
     \right.
\end{equation}
We denote by $\lambda^{\nu}_{k,D}( W)$ the $k$-th eigenvalue of the above Dirichlet problem. For any subset $A\subset W$, the associated first Dirichlet eigenvalue can be characterized as
$$\lambda^{\nu}_{1,D}(A):=\inf\left\{\frac{D_A(f)}{\langle f,f\rangle_A}:0\neq f\in\R^{\overline{ W}}\, \text{and}\, f(x)=0, \forall x\in\overline{ W}\setminus A\right\},$$
where $\langle\cdot,\cdot\rangle_{A}$ is the inner product w.r.t the measure $\nu(\cdot)$.

Set $N:=\sharp W$. For any $k\in[N]$, define
\begin{equation}\label{eq:eqpp10}\Gamma_k( W):=\min_{(A_1,\cdots,A_k)\in\mathcal{A}_k( W)}\max_{l\in[k]}\lambda^{\nu}_{1,D}(A_l).\end{equation}

Let $f_1,f_2,\cdots,f_k$ be the first $k$ orthonormal eigenfunctions of the Dirichlet eigenvalue problem \eqref{Dirichlet}. Consider the following mapping
\begin{equation}\label{4inducedmap}\aligned
F:  W & \to\R^{k},
\\
x  &\mapsto (f_1(x),f_2(x),\cdots,f_k(x)).
\endaligned\end{equation}
We denote by $\parallel\cdot\parallel$ the Euclidean norm and $\langle\cdot,\cdot\rangle$ the inner product of vectors in $\R^k$. Observe that by the Rayleigh quotient characterization of $\lambda^{\nu}_{k,D}( W),$
$$\frac{\sum_{e=\{x,y\}\in E( W,\overline{ W})}\mu_{xy}\|F(x)-F(y)\|^2}{\sum_{x\in W}\nu(x)\|F(x)\|^2}\leq
\lambda^{\nu}_{k,D}( W).$$
We denote by $\widetilde{S}_{F}$ the support set of $F$, i.e.
$$\widetilde{S}_{F}:=\{x\in W:F(x)\neq {0}\}.$$
Consider the map induced by $F$
$$\widetilde{F}:\widetilde{S}_{F}\mapsto\mathbb{S}^{k-1}, \quad x\mapsto\frac{F(x)}{\|F(x)\|}.$$
For any $x,y\in\widetilde{S}_{F}$, the distance between them is defined as
$$d_{\widetilde{F}}(x,y):=\left\|\frac{F(x)}{\parallel F(x)\parallel}-\frac{F(y)}{\parallel F(y)\parallel}\right\|.$$

\subsection{Spreading lemma}

\begin{defi}\label{defspreading}
Let $F$ be the map defined in \eqref{4inducedmap}. For any $r>0$, $\delta>0$,  $F$ is called $(r,\delta)$-spreading, if for any subset $S\subset W$ with $\mathrm{diam}(S\cap\widetilde{S}_F,d_{\widetilde{F}})\leq r$, one has
$$\sum_{x\in S}\nu(x)\|F(x)\|^2\leq\delta\sum_{x\in W}\nu(x)\|F(x)\|^2.$$
\end{defi}
Following \cite{Trevsan2014,Liu2015}, we prove the following spreading lemma.
\begin{lemma}\label{spreadinglemma}
If $0<r<1$ and $S\subset W$ is a subset satisfying $\text{diam}(S\cap\widetilde{S}_{F},d_{\widetilde{F}})\leq r$, then we have
$$\sum_{x\in S}\nu(x)\|F(x)\|^2\leq\frac{1}{k(1-r^2)}\sum_{x\in W}\nu(x)\|F(x)\|^2.$$
\end{lemma}

\begin{proof}
For any unit vector $z$ in $\R^k$, we have
\begin{eqnarray}\label{eq:eqq3}
&&\sum_{x\in W}\nu(x)\langle z,F(x)\rangle^2=\sum_{x\in W}\nu(x)\left(\sum^{k}_{i=1}z_if_i(x)\right)^2 \\
&=&\sum^{k}_{i,j=1}z_iz_j\sum_{x\in W}f_i(x)f_j(x)\nu(x)=\sum^{k}_{i=1}z^2_i=1.\nonumber
\end{eqnarray}
By calculation,
\begin{eqnarray}\label{eq:eqq1}
\quad\sum_{x\in W}\nu(x)\|F(x)\|^2=\sum_{x\in W}\sum^{k}_{i=1}\nu(x)f^2_i(x)
=\sum^{k}_{i=1}\sum_{x\in W}\nu(x)f^2_i(x)=k.
\end{eqnarray}
For any $y\in S\cap\widetilde{S}_F$, by choosing $z=\frac{F(y)}{\|F(y)\|}$ in \eqref{eq:eqq3}, we obtain
\begin{eqnarray}\label{eq:eqq2}
1&=&\sum_{x\in W}\nu(x)\langle F(x),\frac{F(y)}{\|F(y)\|}\rangle^2=\sum_{x\in W\cap\widetilde{S}_{F}}\nu(x)\langle F(x),\frac{F(y)}{\|F(y)\|}\rangle^2 \nonumber\\
&=&\sum_{x\in W\cap\widetilde{S}_{F}}\nu(x)\|F(x)\|^2\langle \frac{F(x)}{\|F(x)\|},\frac{F(y)}{\|F(y)\|}\rangle^2 \nonumber\\&=&\sum_{x\in W\cap\widetilde{S}_{F}}\nu(x)\|F(x)\|^2(1-\frac{1}{2}d^2_{\widetilde{F}}(x,y))^2 \nonumber\\
&\geq&\sum_{x\in S\cap\widetilde{S}_{F}}\nu(x)\|F(x)\|^2(1-\frac{1}{2}r^2)^2\geq (1-r^2)\sum_{x\in S}\nu(x)\|F(x)\|^2.
\end{eqnarray} The lemma follows from \eqref{eq:eqq1} and \eqref{eq:eqq2}.
\end{proof}

\begin{rem}
By Definition \ref{defspreading},  the map $F$ is $(r,\frac{1}{k(1-r^2)})$-spreading.
\end{rem}

\subsection{Localization lemma}

\
\noindent\\

The $\epsilon$-neighborhood of $S\subset\widetilde{S}_{F}$ with respect to $d_{\widetilde{F}}$ is defined as
$$N_{\epsilon}(S,d_{\widetilde{F}}):=\{x\in W:d_{\widetilde{F}}(x,S)<\epsilon\}.$$
For any subset $S\subset W$, we define
\begin{eqnarray}\label{cut-off}
\theta(x)=\left\{
       \begin{array}{ll}
        0,& \text{if\, } F(x)=0, \\
        \max\{0, 1-\frac{d_{\widetilde{F}}(x,S\cap\widetilde{S}_{F})}{\epsilon}\},& \text{otherwise}.
               \end{array}
     \right.
\end{eqnarray}
The so-called localization of $F$ on the subset $S$ is defined as
\begin{eqnarray}\label{4cut-off}
\Psi:=\theta\cdot F: W\mapsto\R^k.
\end{eqnarray}
It is obvious that $\Psi|_S=F|_S$ and $\mathrm{supp}(\Psi)\subset N_{\epsilon}(S\cap\widetilde{S}_{F},d_{\widetilde{F}})$. We have the following localization lemma.

\begin{lemma}\label{localizationlemma}
For $0<\epsilon<2$, let $\Psi$ be the localization defined in \eqref{4cut-off}. Then for any $e=\{x,y\}\in E( W,\overline{ W})$, we have
\begin{eqnarray}\label{localization}\|\Psi(x)-\Psi(y)\|\leq(1+\frac{2}{\epsilon})\|F(x)-F(y)\|.\end{eqnarray}
\end{lemma}

\begin{proof} The result \eqref{localization} is trivial for $F(x)=F(y)=0.$ If only one of $F(x)$ and $F(y)$ vanishes, then \eqref{localization} follows from the fact that $|\theta|\leq 1$. Hence it suffices to consider the case that $x,y\in\widetilde{S}_{F}$, i.e. for $e=\{x,y\}\in E(\widetilde{S}_{F},\widetilde{S}_{F}).$ In this case, we have
\begin{eqnarray*}
\|\Psi(x)-\Psi(y)\|&=&\|\theta(x)F(x)-\theta(y)F(y)\|\nonumber\\
&\leq&|\theta(x)|\|F(x)-F(y)\|+|\theta(x)-\theta(y)|\|F(y)\|\nonumber\\
&\leq&\|F(x)-F(y)\|+\frac{d_{\widetilde{F}}(x,y)}{\epsilon}\|F(y)\|.
\end{eqnarray*}
Note that
\begin{eqnarray*}
d_{\widetilde{F}}(x,y)\|F(y)\|&=&\left\|\frac{F(x)}{\|F(x)\|}-\frac{F(y)}{\|F(y)\|}\right\|\|F(y)\|\nonumber\\
&=&\left\|\frac{\|F(y)\|}{\|F(x)\|}F(x)-F(y)\right\|\nonumber\\
&\leq&\left\|\frac{\|F(y)\|}{\|F(x)\|}F(x)-F(x)\right\|+\|F(x)-F(y)\|\nonumber\\
&\leq&2\|F(x)-F(y)\|.
\end{eqnarray*}
This proves the lemma.
\end{proof}

\subsection{Some results on random partitions}
\
\noindent\\

Let $(X,d)$ be a metric space. For any $x\in X$, $r>0,$ we denote by $B(x,r):=\{y\in X: d(y,x)\leq r\}$ the ball of radius $r$ centered at $x.$ The metric doubling constant $\rho_X$ of $(X,d)$ is defined as

$$\rho_X:=\inf\left\{c\in\mathbb{N}:\forall x\in X, r>0, \exists\, x_1,\cdots,x_c\in X,
\text{ such that }B(x,r)\subset\bigcup^{c}_{i=1}B(x_i,\frac{r}{2})\right\}.$$
The metric doubling dimension of $(X,d)$ is defined as
$$\mathrm{dim}_d(X):=\log_2\rho_X.$$

A Borel measure $\mu$ on $(X,d)$ is called a doubling measure if there exists a finite number $C_{\mu}$ such that for any $x\in S, r>0$,
$$0<\mu\left(B(x,r)\right)\leq C_{\mu}\mu\left(B\left(x,\frac{r}{2}\right)\right)<+\infty.$$
Similarly, the measure doubling dimension is defined as
$$\mathrm{dim}_{\mu}(X):=\log_2(C_{\mu}).$$

The two doubling dimensions are related by the following lemma, see \cite[p. 67]{Coifman1971}.
\begin{lemma}\label{lemma41}
If a metric space $(X,d)$ has a doubling measure $\mu$, then
$$\mathrm{dim}_d(X)\leq4\mathrm{dim}_{\mu}(X).$$
\end{lemma}

One can check that $d(x,y):=\|x-y\|$ is a metric on $\mathbb{S}^{k-1}$ and we have the following property.

\begin{prop}\label{diameter}
For the metric space $(\mathbb{S}^{k-1},d)$, we have

$\bullet \,\mathrm{diam}(\mathbb{S}^{k-1},d)=2;$

$\bullet \,\mathrm{dim}_d(\mathbb{S}^{k-1})\leq 4(k-1)\log_2\pi.$
\end{prop}
\begin{rem}
For the proof of Proposition \ref{diameter}, one can refer to \cite[Corollary 3.12]{LeeNaor2005}.
\end{rem}

A partition of $(X,d)$ is a map $P:X\rightarrow 2^X$, such that $P(x)$ is the unique set in $\{S_i\}^{m}_{i=1}$ that contains $x$, where $S_i\cap S_j=\emptyset$, $\forall i\neq j$, and $X=\cup^{m}_{i=1}S_i$. We denote by $\mathcal{P}(X)$ the collection of partitions of $(X,d)$.

\begin{defi}[Random partition]
Any probability distribution $\varpi$ on $\mathcal{P}(X)$ is called a random partition of $(X,d).$
\end{defi}
We denote by $\mathrm{supp}(\varpi):=\{P\in\mathcal{P}|\varpi(P)\neq0\}$ the support set of random partition $\varpi$.

\begin{theorem}[see e.g. Theorem 2.4 in \cite{Liu2015}]\label{randompartition}
Let $(X,d)$ be a finite metric subspace of $(Y,d)$. Then for any $r>0$, $\delta\in(0,1)$ there exists a random partition $\varpi,$  such that

$\bullet$ $\text{for any } P\in\mathrm{supp}(\varpi)$, any $S$ in the partition $P$, one has $\mathrm{diam}(S)\leq r$;

$\bullet \,$ for any $x$, one has $\mathbb{P}_{\varpi}[B(x, \frac{r}{\alpha})\subset P(x)]\geq1-\delta$, where $\alpha=\frac{32\mathrm{dim}_d(Y)}{\delta}.$
\end{theorem}
\begin{rem}
A random partition obtained in the above theorem is called an $(r,\alpha,1-\delta)$-padded random partition.
\end{rem}
\begin{rem}
Random partition theory was firstly developed in theoretical computer science and has many import applications in pure mathematics, see \cite{Gupta2003,LeeNaor2005,Trevsan2014,Liu2015}.
\end{rem}

\begin{lemma}[see e.g. Lemma 6.2 in \cite{Liu2015}]\label{partition}
Let $0<r<1$, $\alpha>0$, $k\in\mathbb{N}^+$. Suppose that $F$ is $(r,\frac{1}{k}(1+\frac{1}{8k}))$-spreading, and there exists a $(r,\alpha,1-\frac{1}{4k})$-random partition on $(\widetilde{S}_{F},d_{\widetilde{F}})$, then there exist $k$ non-empty disjoint subset   $T_1,T_2,\cdots,T_k\subset\widetilde{S}_{F}$ such that

$\bullet\quad d_{\widetilde{F}}(T_i,T_j)\geq2\frac{r}{\alpha}, \, \forall 1\leq i\neq j\leq k;$

$\bullet\quad \sum_{x\in T_i}\nu(x)\|F(x)\|^2\geq\frac{1}{2k}\sum_{x\in W}\nu(x)\|F(x)\|^2, \, \forall 1\leq i\leq k.$
\end{lemma}

\subsection{Main results of this section}
\begin{theorem}\label{hd}
Let $(V,\mu,\nu)$ be a weighted graph and $ W\subset V$ be a finite subset. For the Dirichlet eigenvalue problem \eqref{Dirichlet} on $ W$, we have
$$\lambda^{\nu}_{k,D}( W)\geq\frac{c}{k^6}\Gamma_k( W),$$ where $\Gamma_k(W)$ is defined in \eqref{eq:eqpp10}.

\end{theorem}
\begin{proof}
Choosing $r=\frac{1}{3\sqrt{k}}$, $F$ is $(r,\frac{1}{k}(1+\frac{1}{8k}))$-spreading  by Lemma \ref{spreadinglemma}. If we further take $\delta=\frac{1}{4k}$, then $\widetilde{S}_{F}$ has an $(r,\alpha,1-\frac{1}{4k})$-padded random partition by Theorem \ref{randompartition} with
$$\alpha=128k\mathrm{dim}_d(\mathbb{S}^{k-1}).$$
From Proposition \ref{diameter}, we know that $\alpha\leq128Ck(k-1)$, where $C=4\log_2\pi$. Then by Lemma \ref{partition}, we can find $k$ disjoint subsets $T_1, T_2, \cdots, T_k$, such that

$\bullet \quad d_{\widetilde{F}}(T_i,T_j)\geq2\frac{r}{\alpha}\geq\frac{2}{3\sqrt{k}}\frac{1}{128Ck(k-1)},\,\forall 1\leq i\neq j\leq k;$

$\bullet\quad \sum_{x\in T_i}\nu(x)\|F(x)\|^2\geq\frac{1}{2k}\sum_{x\in W}\nu(x)\|F(x)\|^2, \, \forall 1\leq i\leq k.$

Let $\{\theta_i\}^k_{i=1}$ be $k$ cut-off functions defined as in \eqref{cut-off}, where $S$ is replaced by $T_i$ and $\epsilon=\frac{1}{3\sqrt{k}}\frac{1}{128Ck(k-1)}$. Similar to \eqref{4cut-off}, we obtain $k$ localizations of $F$ satisfying

$$\Psi_i|_{T_i}=F|_{T_i}, \quad\forall1\leq i\leq k,$$
$$\mathrm{supp}(\Psi_i)\cap\mathrm{supp}(\Psi_j)=\emptyset,\quad \forall1\leq i\neq j\leq k.$$

Applying Lemma \ref{localizationlemma}, for any $1\leq i\leq k,$ we have
\begin{eqnarray*}
&&\frac{\sum_{e=\{x,y\}\in E( W,\overline{ W})}\mu_{xy}\|\Psi_i(x)-\Psi_i(y)\|^2}{\sum_{x\in\mathrm{supp}(\Psi_i)}\nu(x)\|\Psi_i(x)\|^2}\nonumber\\
&\leq&\frac{(1+\frac{2}{\epsilon})^2\sum_{e=\{x,y\}\in E( W,\overline{ W})}\mu_{xy}\|F(x)-F(y)\|^2}{\frac{1}{2k}\sum_{x\in W}\nu(x)\|F(x)\|^2}\nonumber\\
&=&2k(1+768C\sqrt{k}k(k-1))^2\frac{\sum_{e=\{x,y\}\in E( W,\overline{ W})}\mu_{xy}\|F(x)-F(y)\|^2}{\sum_{x\in W}\nu(x)\|F(x)\|^2}\nonumber\\
&\leq&2\times(786C)^2k^6\lambda^{\nu}_{k,D}( W).
\end{eqnarray*}
Write $\Psi_i(x)=(\psi^1_i(x),\psi^2_i(x),\cdots,\psi^k_i(x)).$ Hence for any $1\leq i\leq k,$ there exists a coordinate index $a_i\in\{1,2,\cdots,k\}$ such that $\psi^{a_i}_i$ is not identically zero and
$$\frac{\sum_{e=\{x,y\}\in E( W,\overline{ W})}\mu_{xy}|\psi^{a_i}_i(x)-\psi^{a_i}_i(y)|^2}{\sum_{x\in W}\nu(x)|\psi^{a_i}_i(x)|^2}\leq ck^6\lambda^{\nu}_{k,D}( W),$$
where $c=2\times(786C)^2$. Set $A_i:=\mathrm{supp}(\psi^{a_i}_i),$ for any $1\leq i\leq k.$ Then we have $(A_1,A_2,\cdots,A_k)\in\mathcal{A}_k( W)$ and for any $1\leq i\leq k,$
$$\lambda^{\nu}_{1,D}(A_i)\leq ck^6\lambda^{\nu}_{k,D}( W).$$
Then by the definition of $\Gamma_k( W),$ \eqref{eq:eqpp10}, we have
$$\lambda^{\nu}_{k,D}( W)\geq\frac{c}{k^6}\Gamma_k( W).$$
\end{proof}

\section{Higher order Cheeger estimate for DtN operators}

Let $(V,E,\mu)$ be an infinite graph and ${\Omega}\subset V$ be an infinite subset. For a finite subset $ W\subset\overline{\Omega}$ with $ W\cap\delta\Omega\neq\emptyset,$
let $\sigma_k( W)$ be the $k$-th eigenvalue of the DtN operator on $ W$.

Following the method proposed in \cite{Miclo2017}, we prove higher order Cheeger estimates for the DtN operators. For any $r>0$, consider the following measure defined on $\overline{\Omega}$:
\[
m^{(r)}_x=\left\{
       \begin{array}{ll}
        d(x), &\quad x\in\delta\Omega, \\
        \frac{1}{r}d(x), &\quad x\in\Omega.
               \end{array}
     \right.
\]
For any $f\in\mathbb{R}^{\overline{\Omega}}$, set
\begin{equation}\label{def:blowup}\Delta^{(r)}(f)(x):=\frac{1}{m^{(r)}_x}\sum_{y\in\overline{\Omega}}\mu_{xy}(f(y)-f(x)),\, \forall x\in\Omega.\end{equation}
Let $ W\subset\overline{\Omega}$, $ W\cap\delta\Omega\neq\emptyset$, be a finite subset. Consider the following Dirichlet problem
\begin{equation}\label{dtnapproximate}
\left\{
       \begin{array}{ll}
        -\Delta^{(r)}(f)(x)=\lambda f(x), &\quad x\in W, \\
        f(x)=0, &\quad x\in\delta W.
               \end{array}
     \right.
\end{equation}
We denote by $\lambda^{(r)}_{k,D}( W)$ the $k$-th eigenvalue of the above Dirichlet problem. Recall that $N=\sharp W$. Set $P:=\sharp( W\cap\delta\Omega)$ and we have the following approximation result.

\begin{prop}\label{approximation}
For any $k\in[P]:=\{1,2,\cdots,P\}$, we have
$$\lim_{r\rightarrow +\infty}\lambda^{(r)}_{k,D}( W)=\sigma_k( W)$$
and for any $k\in[N]\setminus[P]$,
$$\lim_{r\rightarrow +\infty}\lambda^{(r)}_{k,D}( W)= +\infty.$$
\end{prop}

\begin{proof}
Let $\Phi^{(r)}_1,\Phi^{(r)}_2,\cdots,\Phi^{(r)}_N$ be the eigenfunctions corresponding to $\lambda^{(r)}_{k,D}( W),k\in[N],$ such that
\begin{eqnarray}\label{orthogonal}\sum_{x\in W}m^{(r)}_x\Phi^{(r)}_i(x)\Phi^{(r)}_j(x)=0,\  \forall r\in(0,+\infty),\forall i\neq j \in[N],\end{eqnarray}
normalized such that
$$\|\Phi^{(r)}_i\|_{\infty}=1,\ \forall r\in(0,+\infty),\forall i\in[N].$$
Consider $s\in[N]$ satisfying
$$\liminf\limits_{r\rightarrow+\infty}\lambda^{(r)}_{s,D}( W)<+\infty,$$
$$\liminf\limits_{r\rightarrow+\infty}\lambda^{(r)}_{s+1,D}( W)=+\infty.$$

For any $k\in[s]$, by the compactness, there exists an increasing sequence of positive numbers $\{r_n\}_{n\in\mathbb{N}}$, a non-negative finite number $\lambda^{(\infty)}_{k,D}( W)$, and a function $\Phi^{(\infty)}_k\in\R^{ W}$ with $\|\Phi^{(\infty)}_k\|_{\infty}=1$, such that
\begin{equation}\label{converge1}\lim_{n\rightarrow \infty}r_n=+\infty,\end{equation}
\begin{equation}\label{converge2}\lim_{n\rightarrow\infty}\lambda^{(r_n)}_{k,D}( W)=\lambda^{(\infty)}_{k,D}( W),\end{equation}
\begin{equation}\label{converge3}\lim_{n\rightarrow\infty}\Phi^{(r_n)}_k=\Phi^{(\infty)}_k.\end{equation}
For $x\in W\cap\delta\Omega$,
$$\Delta(\Phi^{(r_n)}_k)(x)=\Delta^{(r_n)}(\Phi^{(r_n)}_k)(x)=-\lambda^{(r_n)}_{k,D}( W)\Phi^{(r_n)}_k(x).$$
Letting $n\rightarrow\infty$, we get
\begin{equation}\label{dtnEquation}\Delta(\Phi^{(\infty)}_k)(x)=-\lambda^{(\infty)}_{k,D}( W)\Phi^{(\infty)}_k(x).\end{equation}
For $x\in W\cap\Omega$,
$$\Delta^{(r_n)}(\Phi^{(r_n)}_k)(x)=r_n\Delta(\Phi^{(r_n)}_k)(x)=-\lambda^{(r_n)}_{k,D}( W)\Phi^{(r_n)}_k(x).$$
Letting $n\rightarrow\infty$, we have
$$\Delta(\Phi^{(\infty)}_k)(x)=\lim_{n\rightarrow\infty}\frac{-\lambda^{(r_n)}_{k,D}( W)\Phi^{(r_n)}_k(x)}{r_n}=0.$$

Hence $\Phi^{(\infty)}_k$ is the solution of \eqref{existlemma} with $f$ replaced by $\Phi^{(\infty)}_k|_{ W\cap\delta\Omega}$. Note that $\Phi^{(\infty)}_k|_{ W\cap\delta\Omega}\neq0$, otherwise one can conclude that $\Phi^{(\infty)}_k=0$, which contradicts to $\|\Phi^{(\infty)}_k\|_{\infty}=1.$ Hence $\lambda^{(\infty)}_{k,D}( W)$ is an eigenvalue of the DtN operator on $ W$ by \eqref{dtnEquation} and the fact that the DtN operator equals the minus Laplace on $W\cap\delta\Omega$. Letting $r\rightarrow\infty$ in \eqref{orthogonal}, we have
$$\sum_{x\in W\cap\delta\Omega}d(x)(\Phi^{(\infty)}_i|_{ W\cap\delta\Omega})(x)(\Phi^{(\infty)}_j|_{ W\cap\delta\Omega})(x)=0,\quad\forall i\neq j\in[s].$$
This implies $\Phi^{(\infty)}_i|_{ W\cap\delta\Omega}$, $1\leq i\leq s$, are linear independent. Hence $s\leq P.$

Let $\psi_1,\psi_2,\cdots,\psi_P$ be the eigenfunctions of $\Lambda_{ W}$ associated to the eigenvalues $\sigma_1( W),\sigma_2( W),\cdots,\sigma_P( W)$ respectively and $\Psi_i,\forall i\in[P]$, be the solutions of \eqref{existlemma} with $f$ replaced by $\psi_i$. We further assume that $\|\Psi_k\|_{\infty}=1$, $\forall k\in[P]$. Consider the vector space $U$ defined as 
$$U:=\mathrm{span}\{\Psi_k:k\in[P]\}.$$
Note that $\dim U=p$. Then we have for any $r>0$,
$$\lambda^{(r)}_{P,D}( W)\leq\sup_{F\in U\setminus\{0\}}\frac{-\sum_{x\in W}m^{(r)}_xF(x)\Delta^{(r)}(F)(x)}{\sum_{x\in W}m^{(r)}_xF^2(x)}.$$
Since the functions from $U$ are harmonic on $ W\cap\Omega$, we have for any $r>0$,
\begin{eqnarray*}
\lambda^{(r)}_{P,D}( W)&\leq&\sup_{F\in U\setminus\{0\}}\frac{-\sum_{x\in W\cap\delta\Omega}m^{(r)}_xF(x)\Delta F(x)}{\sum_{x\in W}m^{(r)}_xF^2(x)}\nonumber\\
&\leq&\sup_{F\in U\setminus\{0\}}\frac{\sum_{x\in W\cap\delta\Omega}m^{(r)}_xf(x)\Lambda (f)(x)}{\sum_{x\in W\cap\delta\Omega}m^{(r)}_xf^2(x)}\nonumber\\
&\leq&\sup_{F\in U\setminus\{0\}}\frac{\sigma_P( W)\sum_{x\in W\cap\delta\Omega}d(x)f^2(x)}{\sum_{x\in W\cap\delta\Omega}d(x)f^2(x)}\nonumber\\
&=&\sigma_P( W),
\end{eqnarray*}
where $f$ is the restriction of $F$ on $ W\cap\delta\Omega.$
Hence
\begin{eqnarray}\label{sublimitbound} \limsup_{r\rightarrow+\infty}\lambda^{(r)}_{P,D}( W)<+\infty.\end{eqnarray}
This implies $s\geq P$ and finally $s=P$. Hence, for $\{r_n\}_{n=1}^\infty$ satisfying \eqref{converge1}-\eqref{converge3}, we have
$$\lim_{n\rightarrow\infty}\lambda^{(r_n)}_{k,D}( W)=\sigma_k( W),\quad\forall\, k\in[P].$$

Note that for any sequence $r_i\to \infty,$ we can extract a subsequence $\{r_{i_j}\}$ satisfying \eqref{converge1}-\eqref{converge3}. Moreover, by the same argument above, one can show that
$$\lim_{j\rightarrow\infty}\lambda^{(r_{i_j})}_{k,D}( W)=\sigma_k( W),\quad\forall k\in[P].$$ This yields that
$$\lim_{r\rightarrow\infty}\lambda^{(r)}_{k,D}( W)=\sigma_k( W).$$
Then the proposition follows.
\end{proof}

Now we are ready to prove Theorem \ref{hdtn}.

\begin{proof}[Proof of Theorem \ref{hdtn}] For the upper bound estimate,
choose $(A_1,\ldots,A_k)\in\mathcal{A}_k( W)$ that achieves ${h}^k_J( W)$. Consider $H:=\mathrm{span}\{\chi_{A_l}:l\in[k]\}\subset\R^{ W}$. Then $\dim H=k$ and $H|_{ W\cap\delta\Omega}\subset\R^{ W\cap\delta\Omega}.$
Notice that
\begin{eqnarray*}
\frac{D_{ W}(\chi_{A_l})}{\langle \chi_{A_l},\chi_{A_l}\rangle_{ W\cap\delta\Omega}}=\frac{\sum_{e=\{x,y\}\in E( W,\overline{ W})}\mu_{xy}(\chi_{A_l}(y)-\chi_{A_l}(x))^2}{d(A_l\cap\delta\Omega)}
=\frac{\mu(\partial_{ W}A_l)}{d(A_l\cap\delta\Omega)}.
\end{eqnarray*}
Hence by \eqref{min_max_finiDefi},
\begin{eqnarray*}
\sigma_k( W)&\leq&\max_{l\in[k]}\frac{D_{ W}(u^{ W}_{\chi_{A_l}})}{\langle \chi_{A_l},\chi_{A_l}\rangle_{ W\cap\delta\Omega}}=\max_{l\in[k]}\frac{\mu(\partial_{ W}A_l)}{d(A_l\cap\delta\Omega)}={h}^k_{J}( W).
\end{eqnarray*}
Next, we prove the lower bound estimate. For any $k\in[P]$,
\begin{eqnarray*}
\sigma_k( W)&=&\lim_{r\rightarrow+\infty}\lambda^{(r)}_{k,D}( W)
\geq\lim_{r\rightarrow+\infty}\frac{c}{k^6}\min_{(A_1,\cdots,A_k)\in\mathcal{A}_k(W)}\max_{l\in[k]}\lambda^{(r)}_{1,D}(A_l)\nonumber\\
&=&\frac{c}{k^6}\min_{(A_1,\cdots,A_k)\in\mathcal{A}_k(W)}\max_{l\in[k]}\lim_{r\rightarrow+\infty}\lambda^{(r)}_{1,D}(A_l)\nonumber\\
&=&\frac{c}{k^6}\min_{(A_1,\cdots,A_k)\in\mathcal{A}_k(W)}\max_{l\in[k]}\sigma_1(A_l)\nonumber\\
&\geq&\frac{c}{k^6}\min_{(A_1,\cdots,A_k)\in\mathcal{A}_k(W)}\max_{l\in[k]}\frac{1}{2}{h}_J(A_l){h}(A_l)\nonumber\\
&=&\frac{c^{\prime}}{k^6}{h}_k( W),
\end{eqnarray*}
the first inequality follows from Theorem \ref{hd}.
\end{proof}

Finally, by Theorem \ref{hdtn} and Lemma \ref{eigenvalueapproximation}, we can prove Theorem~\ref{infinitehigherlowwerbound}.

\begin{proof}[Proof of Theorem~\ref{infinitehigherlowwerbound}]
For the upper bound estimate,
\begin{eqnarray*}{\sigma}_k(\Omega)=\lim_{i\to \infty}\sigma_k( W_i)\leq\lim_{i\to \infty}{h}^{k}_J( W_i)={h}^{k}_J(\Omega).
\end{eqnarray*}

For the lower bound estimate,
\begin{eqnarray*}
\sigma_k(\Omega)=\lim_{i\to \infty}\sigma_k( W_i)\geq\lim_{i\to \infty}\frac{c}{k^6}{h}_k( W_i)=\frac{c}{k^6}{h}_k(\Omega).
\end{eqnarray*}
Then the theorem follows.
\end{proof}

\bibliography{infiniteDTNproblems}
\bibliographystyle{alpha}

\end{document}